\documentclass{ecai} 
\usepackage{latexsym}
\usepackage{amssymb}
\usepackage{amsmath}
\usepackage{amsthm}
\usepackage{booktabs}
\usepackage{enumitem}
\usepackage{graphicx}
\usepackage{color}
\usepackage{framed}
\usepackage{tikz}

\newtheorem{theorem}{Theorem}
\newtheorem{lemma}[theorem]{Lemma}

\newtheorem{proposition}[theorem]{Proposition}

\newtheorem{definition}{Definition}
\newtheorem{example}{Example}

\newcommand{\BibTeX}{B\kern-.05em{\sc i\kern-.025em b}\kern-.08em\TeX}

\begin{document}

%%%%%%%%%%%%%%%%%%%%%%%%%%%%%%%%%%%%%%%%%%%%%%%%%%%%%%%%%%%%%%%%%%%%%%%%

\begin{frontmatter}

%%% Use this command to specify your submission number.
%%% In doubleblind mode, it will be printed on the first page.

\paperid{1687} 

%%% Use this command to specify the title of your paper.

\title{Optimal Diffusion Auctions}

%%% Use this combinations of commands to specify all authors of your 
%%% paper. Use \fnms{} and \snm{} to indicate everyone's first names 
%%% and surname. This will help the publisher with indexing the 
%%% proceedings. Please use a reasonable approximation in case your 
%%% name does not neatly split into "first names" and "surname".
%%% Specifying your ORCID digital identifier is optional. 
%%% Use the \thanks{} command to indicate one or more corresponding 
%%% authors and their email address(es). If so desired, you can specify
%%% author contributions using the \footnote{} command.
\author{\fnms{Yao}~\snm{Zhang}}
\author{\fnms{Shanshan}~\snm{Zheng}}
\author{\fnms{Dengji}~\snm{Zhao}} 
% \author{\fnms{Yao}~\snm{Zhang}\orcid{0000-0001-6406-5661}}
% \author{\fnms{Shanshan}~\snm{Zheng}\orcid{0000-0002-9768-1660}}
% \author{\fnms{Dengji}~\snm{Zhao}\orcid{0000-0002-9572-1753}\thanks{Corresponding Author. Email: zhaodj@shanghaitech.edu.cn}} 

\address{ShanghaiTech University, Shanghai, China}
% \address[B]{Short Affiliation of Second Author and Third Author}
% \address[C]{Short Alternate Affiliation of Third Author}

%%% Use this environment to include an abstract of your paper.

\begin{abstract}
Diffusion auction design is a new trend in mechanism design for which the main goal is to incentivize existing buyers to invite new buyers, who are their neighbors on a social network, to join an auction even though they are competitors. With more buyers, a diffusion auction will be able to give a more efficient allocation and receive higher revenue. Existing studies have proposed many interesting diffusion auctions to attract more buyers, but the seller's revenue is not optimized. Hence, in this study, we investigate what optimal revenue the seller can achieve by attracting more buyers. Different from the traditional setting, the revenue that can be achieved in a diffusion auction highly relies on the structure of the network. Hence, we focus on optimal auctions with given classes of underlying networks. We propose a class of mechanisms, where for any given structure, an optimal diffusion mechanism can be found. We point out that it implies an idea of "reserve structure". Moreover, we show that an optimal mechanism that handles all structures does not exist. Therefore, we also propose mechanisms that have bounded approximations of the optimal revenue in all structures.
\end{abstract}

\end{frontmatter}

%%%%%%%%%%%%%%%%%%%%%%%%%%%%%%%%%%%%%%%%%%%%%%%%%%%%%%%%%%%%%%%%%%%%%%%%

\section{Introduction}

Single-item auction is a classic mechanism design problem, where a seller sells an item to a group of buyers~\cite{vickrey1961counterspeculation}. Traditionally, the group of buyers is assumed to be known to the seller, and the seller can only maximize efficiency or revenue among them. 
Recently, researchers started to model the connections between buyers and utilize their connections to attract more buyers to join the auction~\cite{zhao2021mechanism}. It is beneficial because we can further improve the seller's revenue. It has been shown that adding one more buyer to a second-price auction is more valuable than using the optimal mechanism among the original buyers~\cite{bulow1996auctions}. To attract more buyers to an auction through their connections, a new approach has been proposed by incentivizing existing buyers to invite their neighbors,
which is called \emph{diffusion auction design}~\cite{li2022diffusion}. The challenge in this new design is that buyers would not invite each other by default as they are competing for the same item, which is called the Referrer's Dilemma in~\cite{jeong2020referrer}. 

% \cite{li2017mechanism} mathematically initiated
The model of diffusion auction design was mathematically initiated in~\cite{li2017mechanism}, which treats the invitation as reporting neighbors and makes a buyer's neighbor set a part of her type. It also demonstrated that the classic VCG mechanism~\cite{vickrey1961counterspeculation,clarke1971multipart,groves1973incentives} can be extended to incentivize buyers to invite their neighbors, but it will give a deficit to the seller. Then, they proposed an incentive-compatible (i.e., all buyers are incentivized to report their valuations truthfully and propagate the auction information to all their neighbors) mechanism 
called the Information Diffusion Mechanism (IDM), which guarantees a higher revenue than holding VCG among the seller's direct neighbors. Later, \citeauthor{LiHZY19Diffusion}~[\citeyear{LiHZY19Diffusion}] further proposed a class of mechanisms called the Critical Diffusion Mechanisms (CDM) that contains IDM. However, these mechanisms are not designed to optimize the seller's revenue and it is not clear what maximal revenue we can achieve under this new setting. 
Therefore, in this paper, we focus on the optimal revenue of the seller. More precisely, given prior distributions of the buyers' valuations, we search for an incentive-compatible (IC) and individually rational (IR) mechanism to maximize the expected revenue.

The well-known optimal mechanism for traditional single-item auction with given prior type distributions, Myerson's mechanism~\cite{myerson1981optimal}, is not IC here as it does not consider the connections between buyers. The main difficulty is that the revenue of a diffusion auction may highly depend on the network structure. Hence, the optimality should also be defined along with structures. Naturally, optimizing the expected revenue over the induced type profile distribution is the ultimate goal. However, in diffusion auctions, buyers’ types are not typically multi-parametric since their neighbor sets are not numerical and independent. Hence, in this paper, we take the expected revenue over the valuation part with given classes of underlying networks. This optimality is also in line with the notion of optimality in traditional settings. There are also other ways to simplify the problem. For example, to define the expected revenue for a specific class of mechanisms (so that the dimension of the structure can be eliminated)~\cite{Bhattacharyya2022OptimalRA}. Therefore, our work is an essential step toward comprehending the issue.

For the above purpose, we propose a class of IC and IR diffusion auctions called the $k$-Partial Winner of Myerson's ($k$-PWM), %Interestingly, for any given network, we find %have found
%an optimal mechanism for its structure only, which is called the $k$-Partial Winner of Myerson's ($k$-PWM),
where $k$ is a parameter associated with the class of underlying structures. %$k$-PWM can achieve the same expected revenue as Myerson's mechanism without considering the network and be incentive-compatible at the same time.
The key component of $k$-PWM is the concept of \emph{potential winners}. A potential winner is a winner under Myerson's mechanism if she does not invite anyone. %Unfortunately, with $k$-PWM, a globally optimal mechanism that can ideally maximize the expected revenue under all network structures does not exist.
By characterizing these special buyers, this class can give an optimal mechanism for any specific class of underlying structures. Although a $k$-PWM gives the optimal revenue for its target structures, it may get zero revenue for the other structures. Moreover and unfortunately, we show that the existence of $k$-PWM implies that an optimal mechanism over all structures does not exist, i.e., no mechanism can have higher expected revenue than others in all underlying structures. Therefore, our next goal is to find a mechanism to approximate the optimal revenue under all structures, which could not be achieved with any of the existing mechanisms.

To better approximate the optimal revenue, we propose a mechanism called the \emph{Closest Winner of Myerson's} (CWM) and a general class called the \emph{CWM with Shifted Reserve Prices} (CWM-SRP). As a base mechanism in this class, %CWM is incentive-compatible and behaves like Myerson's mechanism if there is no cut-point in the network. It is the first mechanism that has a constantly bounded approximation of the optimal revenue under all structures.
CWM simply allocates the item to the potential winner who has the shortest distance to the seller; other variants in the class will give more opportunities to buyers who are far from the seller by increasing the reserve prices of buyers closer to the seller. The CWM also has a closed-form theoretical bounded approximation of the optimal revenue in the worst case among all structures with fixed distributions of the bidders' valuations.

\subsection{Other Related Work}
There is a rich literature following~\cite{li2017mechanism} by proposing different diffusion mechanisms for different settings~\cite{zhao2021mechanism}. For example, \cite{LiHZZ18Customer} considered a customer-sharing model where efficiency can be achieved; \cite{ZhaoLXHJ18Selling} and \cite{KawasakiBTTY20Strategy} extended IDM to a multi-unit setting where each buyer can only bid for one unit. On the other hand, \cite{LiHZY19Diffusion} illustrated a class of diffusion mechanisms with IDM being one member of it, and \cite{ZhangZC20Redistribution} utilized the class and gave a redistribution mechanism for diffusion auction to return the revenue to more buyers. More importantly, \cite{ijcai2020-33} further characterized a sufficient and necessary condition for all incentive-compatible diffusion auctions and showed that a mechanism cannot achieve incentive compatibility, individual rationality, efficiency, and weak budget balance simultaneously. \cite{Lee16Referrals} gave the Multilevel Mechanism that achieves efficiency but sacrifices incentive compatibility (i.e., buyers have diffusion incentives but may not reveal their true valuations). A recent review of these studies can be found in~\cite{DBLP:conf/ijcai/GuoH21}. Different from the above, our work focuses on revenue maximization in diffusion auctions.

%%%%%%%%%%%%%%%%%%%%%%%%%%%%%%%%%%%%%%%%%%%%%%%%%%%%%%%%%%%%%%%%%%%%%%%%

\section{Preliminaries}\label{sec:model}
\subsection{The Model}
We consider an auction where a seller $s$ sells one item on a social network. The social network contains the seller $s$ and a buyer set $N=\{1,2, \dots, n\}$. Each buyer $i\in N$ has a private type of $t_i = (v_i, r_i)$, where $v_i$ is her valuation of the item satisfying $v_i\in [\underline{v}, \overline{v}]$ ($\underline{v}$ and $\overline{v}$ are public), and $r_i \subseteq N\setminus \{i\}$ is the set of all her direct neighbors. Let $r_s\subseteq N$ represent the direct neighbors of the seller. The seller or a buyer can only diffuse the information about the auction to her direct neighbors. Let $T_i$ be the type space of buyer $i$ and $T =\times_{i\in N} T_i$ be the type space of all buyers. In an auction mechanism, each buyer $i$ is asked to report her type, and her report is denoted by $t_i' = (v_i', r_i')$. A buyer $i$ can report any valuation $v_i'\in [\underline{v}, \overline{v}]$ and any subset of her neighbors $r_i'\subseteq r_i$. Denote the report profile of all buyers by $t' = (t_1', \dots, t_n')$. Let $t_{-i}'$ be the report profile of all buyers except for $i$, and $t'$ can be written as $(t_i', t_{-i}')$. 

We focus on the scenario where the seller wants to promote the sale in the network. Initially, the seller only knows her neighbors $r_s$, %does not know the entire social network,
and can only notify them of the auction. %buyers in $r_s$ of the auction.
The goal is to incentivize the buyers who are aware of the sale to further invite their neighbors to join the auction. Eventually, the seller can sell the item to those who are finally informed of the auction. Hence, only buyers who can be reached by the seller via a sequence of neighbor declarations are valid buyers, whereas reporting the neighbor set is equivalent to inviting them in practice. In reality, buyers who are not informed of the sale cannot participate in the sale. We keep their reports in the notations only to simplify the definitions.

\begin{definition}
Given a report profile $t'$, we call buyer $i$ a \textbf{valid buyer} if there exists a sequence of buyers $(i_1, i_2, \dots, i_k)$ satisfying that $i_1\in r_s$, $i_j\in r_{i_{j-1}}'$ for $1<j\leq k$, and $i \in r_{i_k}'$. Let $V(t')$ be the set of all valid buyers in the report profile $t'$.
\end{definition}

\begin{definition}
An \textbf{auction mechanism} $M$ is composed of an allocation policy $\pi = \{\pi_i\}_{i\in N}$ and a payment policy $p = \{p_i\}_{i\in N}$, where $\pi_i: T \rightarrow \{0,1\}$ and $p_i: T \rightarrow \mathbb{R}$ are the allocation and payment for $i$ respectively. A mechanism $M = (\pi, p)$ is a \textbf{diffusion auction mechanism} if it satisfies that for all report profiles $t'\in T$,
\begin{itemize}
    \item for all buyers $i\in N\setminus V(t')$, $\pi_i(t') = 0$ and $p_i(t') = 0$;
    \item for all buyers $i\in V(t')$, $\pi_i(t')$ and $p_i(t')$ are independent of the reports of buyers in $N\setminus V(t')$.
\end{itemize}
\end{definition}

In the allocation policy, $\pi_i(t') = 1$ represents that the item is allocated to $i$ while $\pi_i(t') = 0$ represents the opposite. Since we only have one item to sell, we require $\sum_{i\in N} \pi_i(t') \leq 1$ for all $t'\in T$. Given a report profile $t'$ and a buyer $i$'s true type $t_i = (v_i, r_i)$, her utility in $M=(\pi, p)$ is
\[ u_i(t_i, t', (\pi, p)) = \pi_i(t')v_i - p_i(t'). \]
Then, one required property for the mechanism is that each buyer's utility will be non-negative when she reports her type truthfully, i.e., buyers are not forced to join the mechanism.

\begin{definition}
A diffusion auction mechanism $(\pi, p)$ is individually rational (IR) if for all $i\in N$, $t_i\in T_{i}$ and $t'_{-i} \in T_{-i}$,
$ u_i(t_i, ((t_i,t'_{-i}), (\pi, p)) ) \geq 0$,
where $T_{-i}$ is the space of $t'_{-i}$.
\end{definition}

In the next property, we want to ensure that for all buyers, reporting their true type is a dominant strategy.

\begin{definition}
A diffusion auction mechanism $(\pi, p)$ is incentive compatible (IC) if for all $i\in N$, $t_i, t_i'\in T_{i}$ and $t'_{-i} \in T_{-i}$, $u_i(t_i, (t_i, t'_{-i}), (\pi, p)) \geq u_i(t_i, (t'_i, t'_{-i}), (\pi, p))$.
% \[ u_i(t_i, (t_i, t'_{-i}), (\pi, p)) \geq u_i(t_i, (t'_i, t'_{-i}), (\pi, p)). \]
\end{definition}

In this paper, our goal is to find an IR and IC diffusion mechanism that maximizes the seller's expected revenue given the prior distributions of agents' valuations and specific classes of structures. %The expected revenue may also be affected by the structure of the network,
An underlying structure can be defined by the neighbor sets $r_1$, $r_2$, $\dots$, $r_n$ including $r_s$. Denote $r=(r_s,r_1,r_2,\dots , r_n)$ as the \textbf{structure profile}. Let $R_k$ be the space of all connected structure profiles of $k$ buyers (i.e., at most $k$ valid buyers by their report), and $R = \bigcup_{k\in \mathbb{N}^*} R_k$ be the space of all connected structure profiles. Given the structure profile $r$ and agents' valuations $v$, the revenue of an IC and IR mechanism $M = (\pi, p)$ is $rev^r_M(v) = \sum_{i\in N} p_i((v_i, r_i)_{i\in N})$. Then, we define optimal diffusion auctions as follows.

\begin{definition}
% \textcolor{blue}{
Given a set of structure profiles $S \subseteq R$, An IC and IR diffusion auction mechanism $M=(\pi, p)$ is \textbf{optimal over $S$} if for any structure profile $r \in S$, and any other IC and IR mechanism $M'$, %it maximizes the expected revenue 
\[ \mathbb{E}_{\{v_i\}\sim \{F_i\}} [rev^r_M(v)] \geq \mathbb{E}_{\{v_i\}\sim \{F_i\}} [rev^r_{M'}(v)]  \]
% \[ \mathbb{E}^M_{\{v_i\}\sim \{F_i\}} [rev^r] = \mathbb{E}^M_{\{v_i\}\sim \{F_i\}} \left[\sum_{i\in N} p_i((v_i, r_i)_{i\in N})\right] \]
where the buyers' valuations are drawn independently from the prior distributions $\{F_i\}_{i\in N}$\footnote{We allow heterogeneous prior distributions for theoretical completeness, but in reality, since the seller does not know the existence of buyers, identical priors may more common in applications.}. $F_i$ is the cumulative distribution function (c.d.f.) over $[\underline{v}, \overline{v}]$.
\end{definition}

Briefly, given any structure profile $r \in S$, if the expected revenue of an IC and IR mechanism $M$ is higher than that of any other IC and IR mechanisms, then $M$ is optimal over $S$. An ideal mechanism can be optimal over $R$.

\subsection{Myerson's Mechanism}
In single-item auction design without diffusion, Myerson has proposed an optimal solution for revenue~\cite{myerson1981optimal}. We describe it here and it will be used in our design. It first defines the virtual bids of buyers.

\begin{definition}
For any buyer $i$, if the prior distribution of her valuation is $F_i$ and she reports $v_i'$, %then 
define her virtual bid as
\[ \tilde{v}_i = \phi_i(v_i') = v_i' - (1-F_i(v_i'))/f_i(v_i') \]
where $f_i$ is the probability density function (p.d.f.) of $F_i$.
\end{definition}

\begin{framed}
\noindent \textbf{Myerson's Mechanism}

\noindent \rule{\textwidth}{0.5pt}

\noindent \textsc{Input:} a set of buyers $N$ and %their 
valuation reports $\{v_i'\}$.
\begin{enumerate}
% \item For each buyer $i$, calculate her virtual bid $\tilde{v}_i$.
\item Denote the buyer with the highest non-negative virtual bid (with lexicographic tie-break) by $j$ and set $\pi_j = 1$ and $\pi_{i\neq j} = 0$; if such a buyer does not exist, set $\pi_i=0$ and $p_i=0$ for all $i$ and \textsf{goto} \textsc{Output}.
\item Set $p_j = \phi_j^{-1}(\max\left\{0, \max_{i\neq j} \tilde{v}_i \right\})$ and $p_{i\neq j} = 0$.
% \item Set $p_j=\phi_j^{-1}(\tilde{p}_j)$ and $p_{i\neq j} = 0$.
\end{enumerate}
\noindent \textsc{Output:} the allocation $\pi$ and the payment $p$.
\end{framed}

Intuitively, the allocation of Myerson's mechanism maximizes social welfare under virtual bids, and the payment is the lowest bid for the winner to keep winning. Under the assumption of monotone non-decreasing hazard rates ($f_i(z)/(1-F_i(z))$), Myerson's mechanism is IR, IC, and optimal in the fixed buyer set setting~\cite{myerson1981optimal}.% (maximizing the expected revenue)~\cite{myerson1981optimal}.

\begin{lemma}\label{lem:mm}
If for any $i$, the hazard rate of her valuation distribution $f_i(z)/(1-F_i(z))$\footnote{In the rest of the paper, we always assume monotone non-decreasing hazard rates without explicitly stating it for convenience.} is monotone non-decreasing, then Myerson's mechanism is IR, IC and optimal in the single-item auction without diffusion.
\end{lemma}

It is not hard to observe that directly applying Myerson's mechanism as a diffusion auction mechanism is not IC because one can make her competitors absent from the auction by not inviting them. Although it comes to a failure, the revenue of Myerson's mechanism gives us an upper bound (which may not be tight now) of the expected revenue.

\begin{proposition}\label{prop:upperbound}
For all diffusion auctions with a seller $s$ and $n$ valid buyers whose valuations are drawn from %some 
prior distributions independently, the seller's expected revenue under all IC and IR diffusion mechanisms will not exceed %the seller's expected revenue 
that under Myerson's mechanism with $n$ buyers without diffusion.%all $n$ buyers belonging to $r_s$.
\end{proposition}

\begin{proof}
This statement can be proved by contradiction. Given prior distributions, if there is an IC and IR diffusion auction mechanism $M$ and such that the seller's expected revenue in at least one structure profile in $R_n$ exceeds the expected revenue of applying Myerson's mechanism among $n$ buyers in the setting without diffusion, then we can construct a mechanism $M'$ without diffusion as follows. For $n$ valid buyers, $M'$ enumerates all possible structure profiles $r \in R_n$.
%of $n$ valid buyers in the diffusion setting. 
It records the allocation and payment policies when applying $M$ in each %structure profile 
$r$ and finally sets the allocation and payment policies to the same as the one with the highest expected revenue for given prior distributions. %in all structure profiles when applying $M$.
Because the space of all structure profiles is finite for a given set $N$, the enumeration has an end.

Then, $M'$ is still IC and IR for auction without diffusion since buyers cannot affect the structure profiles enumerated and for each structure profile, $M$ is IC and IR. According to the property of $M$, $M'$ will have higher expected revenue than Myerson's mechanism, %However, %according to Lemma~\ref{lem:mm},
%Myerson's mechanism is optimal,
which leads to a contradiction with Lemma~\ref{lem:mm}. %Hence, %As a result,
Therefore, such a mechanism $M$ cannot exist.
\end{proof}

%%%%%%%%%%%%%%%%%%%%%%%%%%%%%%%%%%%%%%%%%%%%%%%%%%%%%%%%%%%%%%%%%%%%%%%%

\section{Optimal Diffusion Mechanisms}\label{sec:local}
In this section, we will design a class of optimal diffusion mechanisms over $R_k$ for any given $k$. The main difficulty in designing an optimal mechanism for specific structures is that although we only care about the expected revenue at specific structures, the mechanism should work for any structure. Concretely, the mechanism should be IC, which means that \textbf{the outputs on other non-target structures will not give buyers a chance to cheat for higher utilities}. %We will see that for any given structure, an optimal diffusion can be found in the class. Moreover, at the end of this section, we will show that an ideal optimal mechanism over all connected structure profiles $R$ does not exist.

\subsection{Local Myerson's Mechanism}
Before proposing our class of mechanisms, we first examine a trivial solution and some intuition about how to extend it. The idea is to run Myerson's mechanism only among the seller's neighbors. It is an IC and IR diffusion mechanism and is optimal over the structure profiles where $r_s = N$, e.g., the one with two buyers shown in Figure~\ref{fig:2-4buyers}(a). When there are three buyers with the structure shown in Figure~\ref{fig:2-4buyers}(b), it cannot be optimal any longer. For optimizing the revenue in such a case, one may consider a mechanism that only runs Myerson's mechanism among all buyers when there are three buyers or two layers of buyers. Such a method may face the issue of incentive compatibility. For example, in the case shown in Figure~\ref{fig:2-4buyers}(c), buyer 1 can block buyer 3's participation to be the winner. %\footnote{Following we use graphs to visualize structure profiles where nodes represent the seller and buyers, and an edge between $i$ and $j$ means they are neighbors.}

\begin{figure}[!htbp]
\centering
\begin{tikzpicture}[scale=.75]
\filldraw (2,-3) circle (1pt) node[align=center, above] {$s$};
\filldraw (1,-4) circle (1pt) node[align=center, left] {$\tilde{v}_1=4$};
\filldraw (3,-4) circle (1pt) node[align=center, right] {$\tilde{v}_2=2$};
\draw (2,-3) -- (1,-4);
\draw (2,-3) -- (3,-4);
\filldraw (1.7,-4.5) circle (0pt) node[align=center, right] {(a)};

\filldraw (0,-5) circle (1pt) node[align=center, above] {$s$};
\filldraw (-1,-6) circle (1pt) node[align=center, left] {$\tilde{v}_1=4$};
\filldraw (1,-6) circle (1pt) node[align=center, right] {$\tilde{v}_2=2$};
\filldraw (-1,-7) circle (1pt) node[align=center, left] {$\tilde{v}_3=5$};
\draw (0,-5) -- (-1,-6) -- (-1,-7);
\draw (0,-5) -- (1,-6);
\filldraw (-.3,-7.5) circle (0pt) node[align=center, right] {(b)};

\filldraw (5,-5) circle (1pt) node[align=center, above] {$s$};
\filldraw (4,-6) circle (1pt) node[align=center, left] {$\tilde{v}_1=4$};
\filldraw (6,-6) circle (1pt) node[align=center, right] {$\tilde{v}_2=2$};
\filldraw (4,-7) circle (1pt) node[align=center, left] {$\tilde{v}_3=5$};
\filldraw (6,-7) circle (1pt) node[align=center, right] {$\tilde{v}_4=1$};
\draw (5,-5) -- (4,-6) -- (4,-7);
\draw (5,-5) -- (6,-6) -- (6,-7);
\filldraw (4.7,-7.5) circle (0pt) node[align=center, right] {(c)};
\end{tikzpicture}
\caption{Three example networks with 2, 3, and 4 buyers respectively, where buyers' virtual bids are listed beside them.}
\label{fig:2-4buyers}
\end{figure}

We can observe the key is that one can block others' participation to be the highest bidder. We characterize buyers who have such an ability as \emph{potential winners}.

\begin{definition}
Given a report profile $t'$, let $V_{-r_i'}(t')$ be the set of valid buyers if we set $r_i'\leftarrow\emptyset$, and $V_{-i}(t')=V_{-r_i'}(t')\setminus \{i\}$. Then define the set of \textbf{potential winners} (of Myerson's) as
% $ W(\tilde{t}) = \{ i \mid i$ has the highest non-negative virtual bid in $N_{\emptyset}(i)$ with lexicographic tie-break$\} $
% \[ W(\tilde{t}) = \left\{ i \mid i\in \arg\max_{j\in N_\emptyset(i)} \{ \tilde{v}_j \} \text{ and } \tilde{v}_i \geq 0 \right\} \]
$ W(t') = \{ i \mid i\in \arg\max_{j\in V_{-r_i'}(t')} \{ \tilde{v}_j \}$ with lexicographic tie-break and  $\tilde{v}_i \geq 0 \} $,
% \textcolor{blue}{
% $$ W(t') = \{ i \mid i\in \arg\max_{j\in V_{-r_i'}(t')} \{ \tilde{v}_j \} \text{~with lexicographic tie-break and~}  \tilde{v}_i \geq 0 \} $$
where $\tilde{v}_i$ is the virtual bid of $i$, and for all buyers $i\in W(t')$, define the \textbf{potential payment} (of Myerson's) %for $i$ 
% as $$p_i^*(t') = \phi_i^{-1}(\max\{0, \max_{j\in V_{-i}(t')} \tilde{v}_j \})$$ %}
as $p_i^*(t') = \phi_i^{-1}(\max\{0, \max_{j\in V_{-i}(t')} \tilde{v}_j \})$.
% \[ p_i^*(t') = \phi_i^{-1}(\arg\min_\nu \{ i\in W(((\nu, r_i'), \tilde{t}_{-i})) \}) \]
\end{definition}

Intuitively, a potential winner has the highest non-negative virtual bid among the buyers without her invitations.
If there are multiple potential winners, there exists a structural relationship among them which is defined as follows.

\begin{definition}[\cite{li2017mechanism}]
Given a report profile $t'$, and the corresponding reported structure profile $r'$,
for a valid buyer $i\in V(t')$, define the \textbf{critical buyers} of $i$ as 
% \textcolor{blue}{
% $$ C(i) = \{j\mid j\in V(t') \text{~and~} j \text{~exists in all simple paths from~} s \text{~to~} i \text{~in~} r' \} $$%}
$C(i) = \{j\mid j\in V(t')$ and all simple paths from $s$ to $i$ in $r'$ pass $j \}$.
% \[ C(i) = \{j\mid j\in N \text{ and } j \text{ exists in all simple paths from } s \text{ to } i \text{ in } r' \} \]
\end{definition}

\begin{lemma}\label{lem:chain}
Given a report profile $t'$, if the size of the potential winner set $|W(t')| \geq 2$, then for all $i,j\in W(t')$, $i\neq j$, we have either $i\in C(j)$ or $j\in C(i)$.%\footnote{W.l.o.g., we assume that there are no identical bids.}.
\end{lemma}

\begin{proof}
We prove it by contradiction. If neither $i\in C(j)$ nor $j\in C(i)$, then we have both $i\in V_{-r_j'}(t')$ and $j\in V_{-r_i'}(t')$. This implies that $\tilde{v}_i < \tilde{v}_j$ and $\tilde{v}_j < \tilde{v}_i$ (or $\tilde{v}_j = \tilde{v}_i$ but $i$ and $j$ both have larger lexicographic order than the other)\iffalse\footnote{W.l.o.g., we assume that there are no identical bids since we use lexicographic tie-break.}\fi, which cannot happen at the same time. Therefore, either $i\in C(j)$ or $j\in C(i)$ is satisfied.
\end{proof}

\subsection{The Partial Winner of Myerson's}
When there is only one potential winner, we can easily decide who is the winner. One idea is to find a unique special potential winner. This suggests the idea of the \emph{$k$-partial potential winner} as follows.

\begin{definition}
Given a report profile $t'$ and a positive integer $k$, a buyer $i$ is a \textbf{$k$-partial potential winner} if there exists $r_i''\subseteq r_i'$ such that $|V(((v_i',r_i''), t_{-i}'))| = k$ and she is the winner of Myerson's mechanism among $V(((v_i',r_i''), t_{-i}'))$. Let the \textbf{$k$-partial potential payment} of the buyer be the corresponding payment in Myerson's mechanism.
\end{definition}

\begin{lemma}
At most one $k$-partial potential winner exists. %in any instance of diffusion auctions.
% There exists at most one $k$-partial potential winner in any instance of diffusion auctions.
\end{lemma}

\begin{proof}
We prove this statement by contradiction. Given a report profile $t'$ and an integer $k$, suppose there are two $k$-partial potential winners $i$ and $j$ such that $i\neq j$. Since we can decrease the number of their neighbors to make them the winner of Myerson's, then $\{i,j\}\subseteq W(t')$. According to Lemma~\ref{lem:chain}, we have either $i\in C(j)$ or $j\in C(i)$. W.l.o.g., suppose $i\in C(j)$. Then, when $i$ reports $r_i''$, $j$ must be excluded from the valid buyers; otherwise, $i$ cannot be the potential winner. Hence, for $j$, the number of valid buyers is no less than $k$ %when $j$
even if she leaves the auction, which means there is no $r_j''\subseteq r_j'$ such that the number of valid buyers becomes $k$. This is a contradiction to that $j$ is a $k$-partial potential winner.
\end{proof}

The above concept %of $k$-partial potential winner
suggests the following mechanism.
\begin{framed}
\noindent \textbf{The $k$-Partial Winner of Myerson's ($k$-PWM)}

\noindent \rule{\textwidth}{0.5pt}

\noindent \textsc{Input:} a set of buyers $N$ and their type report profile $t'$.
\begin{enumerate}
\item Let $m = |V(t')|$. If $m<k$, then let $\pi_i = 0$, $p_i = 0$ for all $i\in N$ and goto \textsc{Output}.
\item If $m=k$, run Myerson's Mechanism among $V(t')$. %all valid agents.
\item If $m>k$, if there exists a $k$-partial potential winner, then let her be the winner, and her payment is her \textbf{minimal} $k$-partial potential payment; otherwise, let $\pi_i = 0$, $p_i = 0$ for all $i\in N\setminus\{s\}$.
\end{enumerate}
\noindent \textsc{Output:} the allocation $\pi$ and the payment $p$.
\end{framed}

Intuitively, $k$-PWM sacrifices the revenue when the number of buyers reached by the seller is less than $k$, and starts the auction when there are at least $k$ buyers. Taking the structures in Figure~\ref{fig:2-4buyers} as examples, a $3$-PWM will run as follows: (i) For the instance of (a), since there are only two buyers, the $3$-PWM does not allocate the item. (ii) For the instance of (b), since there are exactly three buyers, the $3$-PWM runs Myerson's mechanism among them and finally allocates the item to buyer 3. (iii) For the instance of (c), since there are four buyers, and buyer 1 is the only $3$-partial potential winner (she can block buyer 3's participation and then the number of remaining buyers is three and she has the highest virtual bid), then the $3$-PWM allocates the item to buyer 1. %By doing so, it can reach the optimum when the seller reaches exactly $k$ buyers.

\begin{theorem}
$k$-PWM is IR and IC. %incentive compatible and individually rational. %and weakly budget balanced.
\end{theorem}

\begin{proof}
% We show the mechanism is IR and IC respectively.
% \begin{enumerate}
    % \item 
    (1) For IR, %it can be observed that
    in $k$-PWM, a buyer's payment is 0 or the same as Myerson's mechanism, which always ensures non-negative utility when the buyer reports a true valuation.
    % \item
    
    (2) For IC, first, if $i$ misreports her valuation for any report of $r_i'$, she may (i) change nothing, (ii) lose the item with 0 utility, or (iii) win the item with negative utility. It means that for any buyer $i$ and any possible neighbor set $r_i'$ she may report, reporting her truth valuation $v_i$ will maximize her utility. Hence, we only need to check whether a buyer $i$ wants to invite fewer other buyers. Suppose buyer $i$ reports $v_i$ and $r_i' \subset r_i$, then
    % \begin{enumerate}
        % \item
        (i) if the number of valid buyers does not change, buyer $i$'s utility will not change, either;
        % \item
        (ii) if the number of valid buyers decreases, let the number of valid buyers be $m'$ when reporting $r_i'$ and be $m$ when reporting $r_i$. %($m > m'$).
        % \textcolor{blue}{
        \begin{enumerate}
            \item[(a)] If $m'<k$, then the utility of buyer $i$ will be 0.
            \item[(b)] If $m' \geq k$, there are two possibilities. If $i$ was not the $k$-partial potential winner, she still cannot be the winner and gets 0 utility. If $i$ was the $k$-partial potential winner, since we have already given her the minimal $k$-partial potential payment, she cannot have a lower payment. %in this case.
            % \item[(a)] if $m'>k$, then if the original winner is still valid, the allocation and payment of all buyers will not change; otherwise, there is no winner.
            % \item[(b)] if $m'<k$, then the utility of buyer $i$ will be 0. %regardless of what $m$ is. %Hence, buyer $i$ will not get better utility in this case.
            % \item[(c)] if $m'=k$, there are two possibilities. If $i$ was not the $k$-partial potential winner, she still cannot be the winner and gets 0 utility. If $i$ was the $k$-partial potential winner, she can still be the winner. Since we have already given her the minimal $k$-partial potential payment, she cannot have a lower payment in this case. %we will change from running the CWM mechanism to running Myerson's mechanism. If buyer $i$ is the winner initially, her payment is the Myerson's payment when she invites no one else, according to the CWM mechanism. It cannot be greater than her Myerson's payment when she invites $r_i'$ and still becomes the winner. She may also lose the item because Myerson's mechanism consider buyers in $r_i'$. If buyer $i$ is not the winner initially, she cannot be the winner when running Myerson's mechanism, either. %Hence, buyer $i$ will not get better utility in this case.
        % \end{enumerate}
        % Hence, buyer $i$ will not get better utility in all cases.
        \end{enumerate}
        % }
    % $i$ cannot get higher utility in all cases, so $k$-PWM is IC.
    % Taking all the above into account,
    Therefore, we conclude that the $k$-PWM mechanism is IC.
% \end{enumerate}
\end{proof}

\begin{theorem}\label{thm:optimal}
% \textcolor{blue}{
% Given a positive integer $k$,
$k$-PWM is optimal over $R_k$.%}
%$k$-PWM is locally optimal at $r\in R_k$.
\end{theorem}

\begin{proof}
When there are exactly $k$ valid buyers, $k$-PWM runs Myerson's mechanism among all of them. By Proposition~\ref{prop:upperbound}, any other IC and IR mechanism that cannot have a higher expected revenue when there are totally at most $k$ valid buyers. As a result, $k$-PWM is optimal over $R_k$.
\end{proof}

\subsection{Impossibility of Universal Optimality}
% \subsection{Optimal Mechanism over $R$ does not Exist}
Now with the class of $k$-PWMs, we have optimal diffusion auctions for specific classes. One may further ask whether there is a mechanism that can have the highest expected revenue than others on whichever structure. Unfortunately, we can notice that the optimal expected revenue that can be achieved in any structure is the same as Myerson's mechanism with the same number of buyers without diffusion. This gives us a negative answer to the question.

\begin{theorem}\label{thm:global}
An optimal diffusion auction mechanism over $R$ does not exist.
\end{theorem}

\begin{proof}
According to Theorem~\ref{thm:optimal} and Proposition~\ref{prop:upperbound}, if an IC and IR diffusion auction mechanism $M$ is optimal over $R$, $M$ should have the same expected revenue as the $n(r)$-PWMs for any structure profile $r\in R$, where $n(r)$ is the number of buyers in $r$. Since $n(r)$-PWM runs Myerson's mechanism among all buyers in $r$, then $M$ has to always run Myerson's mechanism among all valid buyers (according to the uniqueness of Myerson's solution~\cite{myerson1981optimal,archer2001truthful}), i.e., $M$ directly applies Myerson's mechanism as a diffusion auction mechanism, which is not IC. Hence, such a diffusion mechanism $M$ does not exist.
\end{proof}

\section{Optimal Approximation Mechanisms}\label{sec:general}
Although $k$-PWMs gives the optimal revenue we can achieve in different structures, \textbf{in practice, it is hard to accurately predict how many buyers will attend the auction. It is also not robust as the seller may lose all the revenue when there is only one buyer absent.} Hence, in this section, we consider the approximation ratio to optimal revenue over $R$. %among all structure profiles $R$.

\begin{definition}
An IC and IR diffusion auction mechanism $M$ is $\alpha$-optimal if \[ \inf_{r\in R} \frac{\mathbb{E}_{\{v_i\}\sim\{F_i\}}[rev^r_M]}{\mathbb{E}_{\{v_i\}\sim\{F_i\}}[rev^r_{n(r)\text{-}\mathtt{PWM}}]} \geq \alpha \] where $n(r)$ is the number of valid buyers in $r$.
\end{definition}

We can see that a $k$-PWM itself is not ideal, since its approximation ratio is 0. For existing mechanisms in previous work, IDM~\cite{li2017mechanism} and its extended class CDM~\cite{li2022diffusion} all fail to have a non-zero approximation ratio since they all have 0 revenue if the structure is a single chain. Therefore, we propose a new mechanism that has a bounded positive approximation ratio.

\subsection{The Closest Winner of Myerson's}\label{sec:mechanism}
According to Lemma~\ref{lem:chain}, the potential winner set forms a sequence $(w_1, w_2, \dots, w_m)$, where $w_k\in C(w_{k+1})$ for all $1\leq k<m$. Buyer $w_1$ in the sequence is the closest potential winner to the seller and should be given precedence over all other potential winners because they cannot be informed of the auction without buyer $w_1$'s diffusion. Based on this observation, we propose the following mechanism.

\begin{framed}
\noindent \textbf{The Closest Winner of Myerson's (CWM)}

\noindent \rule{\textwidth}{0.5pt}

\noindent \textsc{Input:} a set of buyers $N$ and their report profile $t'$.
\begin{enumerate}
% \item For each buyer $i\in V(t')$ and her bid $v_i'$, compute her virtual bid $\tilde{v}_i = \phi_i(v_i')$.
\item Find the potential winner set $W(t')$, represented by $\{ w_1,$ $\dots, w_m \}$, where $w_k\in C(w_{k+1})$ for all $1\leq k<m$. If $W(t') = \emptyset$, then set $\pi_i(t') = 0$, $p_i(t') = 0$ for all $i\in N$ and goto \textsc{Output}.
\item Set $\pi_{w_1}(t') = 1$ and $\pi_{i\neq w_1}(t') = 0$.
\item Set $p_{w_1}(t') = p^*_{w_1}(t')$ and $p_{i\neq w_1}(t') = 0$.
\end{enumerate}
\noindent \textsc{Output:} the allocation $\pi$ and the payment $p$.
\end{framed}

Since we only need to find the closest potential winner to the seller, it is not necessary to compute the whole potential winner set and suggests a simple linear-time algorithm.%\footnote{One can find the implementation and corresponding running example in the Appendix A.}.

\begin{framed}
\noindent \textbf{A Simple Algorithm for CWM}

\noindent \rule{\textwidth}{0.5pt}

\noindent \textsc{Input:} a set of buyers $N$ and their report profile $t'$.
\begin{enumerate}
\item Set $B\leftarrow r_s$, and $\tilde{p}_i = 0$ for all $i$.
\item Loop
\begin{enumerate}
\item Run Myerson's mechanism among buyers in the set $B$ and let $(\hat{\pi}, \hat{p})$ be the output.
\item Let $i$ be the winner in $\hat{\pi}$ i.e. $\hat{\pi}_i = 1$ ($i$ may not exist). Set $\tilde{p}_i \leftarrow \max\{ \tilde{p}_i, \hat{p}_i \}$.
\item Set $B\leftarrow B\cup \left( \bigcup_{j\in B, j\neq i} r_j' \right)$. %and $U\leftarrow B$. %For all $j\in B$, $j\neq i$, $B\leftarrow B\cup r_j$.
%\item $B\leftarrow B\setminus U \cup \{i\}$, and then $U\leftarrow U\cup B$.
\end{enumerate}
Until the set $B$ no longer changes.
\item Suppose the output of Myerson's mechanism in the last iteration of the loop is $(\hat{\pi}, \hat{p})$. Let $j$ be the winner in $\hat{\pi}$ i.e. $\hat{\pi}_j = 1$ and then set $\pi_j = 1$ and $\pi_{i\neq j} = 0$; if $j$ does not exist, set $\pi_i = 0$ for all $i$.
%\textcolor{blue}{
For all $i$, set $p_i = \tilde{p}_i \cdot \pi_i$.%}
%\item
% For all $i$, set $p_i = \tilde{p}_i \cdot \mathbb{I}(\pi_i = 1)$, 
% % \[p_i = \tilde{p}_i \cdot \mathbb{I}(\pi_i = 1)\]
% where $\mathbb{I}(\cdot)$ is the indicator function.
\end{enumerate}
\noindent \textsc{Output:} the allocation $\pi$ and the payment $p$.
\end{framed}

Intuitively, buyer $j$ must be the closest buyer in the potential winner set $W(t')$ when the algorithm terminates, because only buyers in $V_{-r_j'}(t')$ are visited in the algorithm and no other buyer in $V_{-r_j'}(t')$ has a chance to be a potential winner. Since the algorithm only runs a breadth-first traversal, %and in each loop, %we can compare each bid of unvisited buyers with the current highest bid in constant time,
CWM can be implemented with linear time. To better present the procedure, we show an example below.

\begin{figure}[!htbp]
\centering
\begin{tikzpicture}[scale=.75]
% \filldraw (-10,-4) circle (1pt) node[align=center, above] {$s$};
% \filldraw (-11,-5) circle (1pt) node[align=center, left] {$\tilde{v}_1=3$};
% \filldraw (-9,-5) circle (1pt) node[align=center, right] {$\tilde{v}_2=2$};
% \filldraw (-12,-6) circle (1pt) node[align=center, left] {$\tilde{v}_3=4$};
% \filldraw (-10,-6) circle (1pt) node[align=center, left] {$\tilde{v}_4=7$};
% \filldraw (-8,-6) circle (1pt) node[align=center, left] {$\tilde{v}_5=1$};
% \filldraw (-10.8,-7) circle (1pt) node[align=center, left] {$\tilde{v}_6=9$};
% \filldraw (-9.2,-7) circle (1pt) node[align=center, left] {$\tilde{v}_7=5$};
% \filldraw (-8.8,-7) circle (1pt) node[align=center, right] {$\tilde{v}_8=6$};
% % \filldraw (-7.2,-7) circle (1pt) node[align=center, below] {$\tilde{v}_9=4$};
% \draw (-10,-4) -- (-11,-5) -- (-12,-6);
% \draw (-10,-4) -- (-9,-5) -- (-10,-6) -- (-10.8,-7);
% \draw (-11,-5) -- (-10,-6);
% \draw (-9,-5) -- (-8,-6) -- (-8.8,-7);
% \draw (-10,-6) -- (-9.2,-7);
% \draw (-8,-6) -- (-7.2,-7);

\filldraw (0,-1) circle (1pt) node[align=center, above] {$s$};
\filldraw (-1,-2) circle (1pt) node[align=center, left] {$\tilde{v}_1=3$};
\filldraw (1,-2) circle (1pt) node[align=center, right] {$\tilde{v}_2=2$};
\filldraw (-2,-3) circle (1pt) node[align=center, left] {$\tilde{v}_3=4$};
\filldraw (0,-3) circle (1pt) node[align=center, left] {$\tilde{v}_4=7$};
\filldraw (2,-3) circle (1pt) node[align=center, left] {$\tilde{v}_5=1$};
\filldraw (-0.8,-4) circle (1pt) node[align=center, left] {$\tilde{v}_6=9$};
\filldraw (0.8,-4) circle (1pt) node[align=center, left] {$\tilde{v}_7=5$};
\filldraw (1.2,-4) circle (1pt) node[align=center, right] {$\tilde{v}_8=6$};
% \filldraw (-7.2,-7) circle (1pt) node[align=center, below] {$\tilde{v}_9=4$};
\draw (0,-1) -- (-1,-2) -- (-2,-3);
\draw (0,-1) -- (1,-2) -- (0,-3) -- (-0.8,-4);
\draw (-1,-2) -- (0,-3);
\draw (1,-2) -- (2,-3) -- (1.2,-4);
\draw (0,-3) -- (0.8,-4);

\filldraw (5,-1) circle (1pt) node[align=center, above] {$s$};
\filldraw[blue] (4,-2) circle (1pt) node[align=center, left] {$1$};
\filldraw (6,-2) circle (1pt) node[align=center, right] {$2$};
\draw (3,-3) circle (1pt) node[align=center, left] {$3$};
\draw (5,-3) circle (1pt) node[align=center, right] {$4$};
\draw (7,-3) circle (1pt) node[align=center, above] {$5$};
\draw (4.2,-4) circle (1pt) node[align=center, left] {$6$};
\draw (5.8,-4) circle (1pt) node[align=center, left] {$7$};
\draw (6.2,-4) circle (1pt) node[align=center, right] {$8$};

\draw (5,-1) -- (4,-2);
\draw[dashed] (4,-2) -- (3,-3);
\draw (5,-1) -- (6,-2);
\draw[dashed] (6,-2) -- (5,-3);
\draw[dashed] (5,-3) -- (4.2,-4);
\draw[dashed] (6,-2) -- (7,-3);
\draw[dashed] (4,-2) -- (5,-3);
\draw[dashed] (7,-3) -- (6.2,-4);
\draw[dashed] (5,-3) -- (5.8,-4);
\filldraw (3,-1.2) circle (0pt) node[align=center, right] {(a)};

\filldraw (0,-5) circle (1pt) node[align=center, above] {$s$};
\filldraw (-1,-6) circle (1pt) node[align=center, above] {$1$};
\filldraw (1,-6) circle (1pt) node[align=center, above] {$2$};
\draw (-2,-7) circle (1pt) node[align=center, left] {$3$};
\filldraw[blue] (0,-7) circle (1pt) node[align=center, right] {$4$};
\filldraw (2,-7) circle (1pt) node[align=center, above] {$5$};
\draw (-.8,-8) circle (1pt) node[align=center, left] {$6$};
\draw (.8,-8) circle (1pt) node[align=center, left] {$7$};
\draw (1.2,-8) circle (1pt) node[align=center, right] {$8$};
% \draw (2.8,-7) circle (1pt) node[align=center, below] {$9$};
\draw (0,-5) -- (-1,-6);
\draw[dashed] (-1,-6) -- (-2,-7);
\draw (0,-5) -- (1,-6) -- (0,-7);
\draw[dashed] (0,-7) -- (-.8,-8);
\draw[dashed] (-1,-6) -- (0,-7);
\draw (1,-6) -- (2,-7);
\draw[dashed] (2,-7) -- (1.2,-8);
\draw[dashed] (0,-7) -- (.8,-8);
% \draw[dashed] (2,-6) -- (2.8,-7);
\filldraw (-2,-5.2) circle (0pt) node[align=center, right] {(b)};

\filldraw (5,-5) circle (1pt) node[align=center, above] {$s$};
\filldraw (4,-6) circle (1pt) node[align=center, above] {$1$};
\filldraw (6,-6) circle (1pt) node[align=center, above] {$2$};
\filldraw (3,-7) circle (1pt) node[align=center, left] {$3$};
\filldraw[blue] (5,-7) circle (1pt) node[align=center, right] {$4$};
\filldraw (7,-7) circle (1pt) node[align=center, above] {$5$};
\draw (4.2,-8) circle (1pt) node[align=center, left] {$6$};
\draw (5.8,-8) circle (1pt) node[align=center, left] {$7$};
\filldraw (6.2,-8) circle (1pt) node[align=center, right] {$8$};
% \filldraw (7.8,-7) circle (1pt) node[align=center, below] {$9$};
\draw (5,-5) -- (4,-6) -- (3,-7);
\draw (5,-5) -- (6,-6) -- (5,-7);
\draw[dashed] (5,-7) -- (4.2,-8);
\draw (4,-6) -- (5,-7);
\draw (6,-6) -- (7,-7) -- (6.2,-8);
\draw[dashed] (5,-7) -- (5.8,-8);
\filldraw (3,-5.2) circle (0pt) node[align=center, right] {(c)};
\end{tikzpicture}
\caption{The above-left graph is an example of the network in a diffusion auction. The following three graphs (a), (b), and (c) show the algorithm to find the closest potential winner. In each step, the blue node represents who has the highest virtual bid among the buyers connected with solid lines.}
\label{fig:example}
\end{figure}
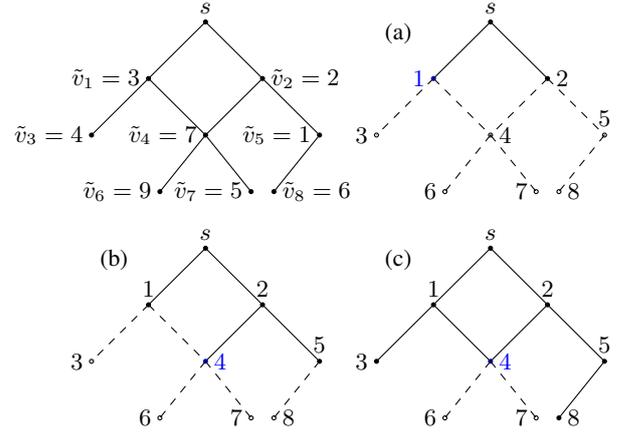

\begin{example}\label{eg:stepwise}
Consider the network shown in Figure~\ref{fig:example} with $n=8$ buyers. The algorithm for CWM runs as follows. First, the auction runs among buyers in $r_s = \{1,2\}$. Since $\tilde{v}_1>\tilde{v}_2$, then buyer 2 has no chance to win and we can introduce buyers in $r_2=\{4,5\}$ to the auction. Repeat the procedure until no more buyers can participate in the auction and we can determine the final winner 4. %In this example, buyer 4 is the final winner.
\end{example}

% Then, we show the properties of CWM.

\begin{theorem}\label{thm:icir}
CWM is IR and IC.
\end{theorem}

\begin{proof}
(i) For IR, if a buyer is not the winner, she has 0 payment. If buyer $i$ is the winner, then her payment is equal to that in Myerson's mechanism running on $V_{-r_i'}(t')$, which ensures non-negative utility when the buyer truthfully reports.

(ii) To see that the mechanism is IC, we first show that for any buyer $i$ and any possible bid $v_i'$ she may report, her diffusion $r_i'$ will not affect her final utility. Then we show that for any buyer $i$ and any possible diffusion $r_i'$, reporting her true valuation $v_i$ will always maximize her final utility. Finally, we can combine them to conclude that truthfully reporting her type $t_i=(v_i,r_i)$ is a dominant strategy.

Consider the case of a buyer $i$ who makes a bid of $v_i'$. If buyer $i$ does not win with her diffusion $r_i'$, she loses the item before considering her diffusion. On the other hand, if the buyer $i$ is the winner, she must have won the competition with all other buyers that can be informed without her diffusion. Hence, 
% For those buyers who must be informed through $i$, they have no chance to compete with her.
her diffusion will not affect her final utility.

Consider the case of a buyer $i$ with diffusion $r_i'$. If she is the winner with the bid $v_i' = v_i$, then she will still be the winner with a bid $v_i'' > v_i$, and her payment will not change since the payment is only determined by other buyers' bids. If she changes bid $v_i'$ to $v_i'' < v_i$, the payment will not change if $i$ still wins, but she may lose the item since her virtual bid will decrease, which makes $i$'s utility be 0. On the other hand, if $i$ is not the winner with a bid $v_i' = v_i$, then she will still have no chance to be the winner with a bid $v_i'' < v_i$, and her utility remains 0. If she changes bid $v_i'$ to $v_i'' > v_i$, the utility will remain 0 if $i$ still loses, but she may win the item since her virtual bid will increase. In this case, $i$'s utility becomes to $v_i - \phi_i^{-1}(\max\{\tilde{\nu}, 0\})$, where $\tilde{\nu} = \max_{j\in V_{-i}(t')} \tilde{v}_j$. Because $i$ is not the winner when reporting $v_i' = v_i$, then $\phi_i(v_i) \leq \max\{\tilde{\nu}, 0\}$, i.e., $v_i - \phi_i^{-1}(\max\{\tilde{\nu}, 0\})\leq 0$. As a result, reporting $v_i$ will always maximize the utility of any buyer $i$ with all possible diffusion $r_i'$.

Combining the above, %for any buyer $i$, truthfully reporting her type $t_i=(v_i,r_i)$ is a dominant strategy, indicating
we can conclude that CWM is IC.
\end{proof}

Then, we can observe two direct propositions. Firstly, we will see CWM is always better than local Myerson's mechanism, which is optimal if we do not utilize diffusion. Secondly, CWM is optimal over some special structures.

\begin{proposition}
CWM's revenue is always no worse than that of Myerson's mechanism among the seller's neighbors.
\end{proposition}

\begin{proof}
If there is no winner in CWM, then there is also no winner in local Myerson's mechanism. If the winner in CWM is one of the seller's neighbors, then she is also the winner in local Myerson's mechanism with fewer competitors. Hence, her payment will not decrease. If the winner in CWM is not one of the seller's neighbors, the lowest bid can make her the winner must be higher than all the seller's neighbors. Therefore, we can conclude the proposition.
\end{proof}

\begin{proposition}\label{prop:all}
CWM is optimal over the set of structure profiles that satisfies $C(i) = \emptyset$ for any buyer $i\in N$.
\end{proposition}

\begin{proof}
Given a structure profile $r$, if for any buyer $i\in N$, $C(i)$ $=$ $\emptyset$, then $V_{-r_i}(t) = N$ for all $i$. Hence, CWM is equivalent to running Myerson's mechanism among $N$. According to Proposition~\ref{prop:upperbound}, the mechanism is optimal %\textcolor{blue}{
over $r$. %the set of structure profiles that have no cut points.%}
%at $r$.
\end{proof}

Note that when social networks are well-connected with a few cut-points, CWM can also achieve a good performance.
More importantly, %This will give us the
we will see it has a bounded gap to optimal revenue for all structure profiles. For simplicity, we assume an identical distribution.

\begin{theorem}~\label{thm:cwm-bound}
If all buyers' valuations are drawn independently from an identical distribution $F$, then CWM is a $(\phi^{-1}(0)/\overline{v})$-optimal mechanism, and the approximation ratio is tight.
\end{theorem}

\begin{proof}
Let $\phi = \phi_1 = \cdots = \phi_n$ and $\tau = \phi^{-1}(0)$. For all structure profiles $r\in R_n$, the expected revenue of CWM is
\begin{align*}
    \mathbb{E}[rev^r_{\mathtt{CWM}}]
    %\mathbb{E}^{\mathtt{CWM}}[rev_s] %& = \mathbb{E}^{\mathtt{CWM}}[rev_s \mid \max_{i\in N} \{ \tilde{v}_i \} \geq 0] \cdot \mathrm{Pr}(\max_{i\in N} \{ \tilde{v}_i \} \geq 0) \\ &\ \  + \mathbb{E}^{\mathtt{CWM}}[rev_s \mid \max_{i\in N} \{ \tilde{v}_i \} < 0] \cdot \mathrm{Pr}(\max_{i\in N} \{ \tilde{v}_i \} < 0) \\
    & = \mathbb{E}[rev^r_{\mathtt{CWM}} \mid \max_{i\in N} \{ \tilde{v}_i \} \geq 0] \cdot \mathrm{Pr}(\max_{i\in N} \{ \tilde{v}_i \} \geq 0) \\
    & \geq \phi^{-1}(0) \cdot \mathrm{Pr}(\max_{i\in N} \{ v_i \} \geq \phi^{-1}(0)) \\
    % & \geq \phi^{-1}(0) \cdot \mathrm{Pr}(v_1 \geq \phi^{-1}(0)) \\
    & = \tau (1-F^n(\tau)).
\end{align*}

On the other hand, we have
\begin{align*}
    \mathbb{E}[rev^r_{n\text{-}\mathtt{PWM}}] & = n(1-F(\tau))F^{n-1}(\tau)\tau \\
    &\
    + n(n-1)\int_\tau^{\overline{v}} F^{n-2}(x)(1-F(x))f(x)x\mathrm{d}x
\end{align*}
where
\begin{align*}
    & n(n-1)\int_\tau^{\overline{v}} F^{n-2}(x)(1-F(x))f(x)x\mathrm{d}x \\ 
    \leq &\ n(n-1)\overline{v}\int_\tau^{\overline{v}} F^{n-2}(x)(1-F(x))f(x)\mathrm{d}x \\
    % \leq &\ n(n-1)\overline{v}\int_\tau^{\overline{v}} F^{n-2}(x)(1-F(x))\mathrm{d}F(x) \\
    = &\ n(n-1)\overline{v}\int_\tau^{\overline{v}} (F^{n-2}(x)-F^{n-1}(x))\mathrm{d}F(x) \\
    = &\ n(n-1)\overline{v} \left( \frac{1}{n}F^n(\tau) - \frac{1}{n-1}F^{n-1}(\tau) + \frac{1}{n(n-1)} \right) \\
    = &\ -\overline{v}n(1-F(\tau))F^{n-1}(\tau) + \overline{v}(1-F^n(\tau))
\end{align*}
% Hence, we have
Therefore, the approximation ratio of CWM is at least
\begin{align*}
    \frac{\mathbb{E}[rev^r_{\mathtt{CWM}}]}{\mathbb{E}[rev^r_{n\text{-}\mathtt{PWM}}]}
    % \geq &\ \frac{\tau(1-F^n(\tau))}{n(1-F(\tau))F^{n-1}(\tau)\tau + \overline{v}(n-1)F^n(\tau) - \overline{v}nF^{n-1}(\tau) + \overline{v}} \\
    & \geq \frac{\tau(1-F^n(\tau))}{n(1-F(\tau))F^{n-1}(\tau)(\tau - \overline{v}) + \overline{v}(1-F^n(\tau))} \\
    & = \frac{\tau}{\frac{-n}{(1/F(\tau))^{n-1} - F(\tau)}(1-F(\tau))(\overline{v} - \tau) + \overline{v}} \\
    & \geq \frac{\tau}{\overline{v}} \quad\quad\quad\quad\quad\quad\quad\quad\quad\quad\quad\quad\quad\quad\quad (n\rightarrow \infty)
\end{align*}

Consider a structure profile $r\in R_n$ in which $r_s = \{1\}$, $r_i = \{i + 1\}$ for all $i<n$, and $r_n = \emptyset$. In this case, %according to the algorithm,
CWM asks buyers $1$, $2$, $\dots$, $n$ one by one whether their valuations are greater than $\phi^{-1}(0)$. The first buyer who has the answer `yes' will win the item and pay $\phi^{-1}(0)$. If all $n$ buyers' bids are less than $\phi^{-1}(0)$, no one will win the item. Hence, the expected revenue when $n$ approaches infinity is exactly $\phi^{-1}(0)$. %$\phi^{-1}(0) \cdot \mathrm{Pr}(\max_{i\in N} \{ v_i \} \geq \phi^{-1}(0))$.
As a result, the approximation ratio is tight. %in Theorem~\ref{thm:cwm-bound} is tight. %among all structure profiles $r\in R_n$, the CWM mechanism has the lowest expected revenue at a chain structure.
\end{proof}

Intuitively, the approximation ratio of CWM is determined by the reserve price. For example, if $F$ is the uniform distribution over $[0,\overline{v}]$, then %the approximation ratio of
CWM is $1/2$-optimal; if $F$ is the Beta distribution $\beta_{2,5}$ over $[0,1]$, then CWM is about $0.236$-optimal. This may look worse for some extreme distributions (e.g., it approaches 0 for $\beta_{2,b}$ when $b$ approaches infinity).

\subsection{Shifted Reserve Prices}
% \textcolor{blue}{
As CWM always chooses the first potential winner to be the winner, % }
%The main bottleneck of CWM is that 
buyers who are away from the seller have fewer opportunities to win the item. We can see it in the following example. 
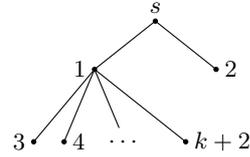
\begin{figure}[htbp]
\centering
\begin{tikzpicture}[scale=.8]
\filldraw (0,0) circle (1pt) node[align=center, above] {$s$};
\filldraw (-1,-.8) circle (1pt) node[align=center, left] {$1$};
\filldraw (1,-.8) circle (1pt) node[align=center, right] {$2$};
\filldraw (-2,-2) circle (1pt) node[align=center, left] {$3$};
\filldraw (-1.5,-2) circle (1pt) node[align=center, right] {$4$};
\filldraw (.5,-2) circle (1pt) node[align=center, right] {$k+2$};
\node (D) at (-.5,-2) {$\cdots$};
\draw (0,0) -- (-1,-.8) -- (-2,-2);
\draw (0,0) -- (1,-.8);
\draw (-1,-.8) -- (-1.5,-2);
\draw (-1,-.8) -- (D);
\draw (-1,-.8) -- (.5,-2);
\end{tikzpicture}
\caption{A structure profile where the majority of buyers in $r_1\setminus r_s$.}
\label{fig:extreme2layers}
\end{figure}
\begin{example}\label{eg:extreme}
Consider the structure profile illustrated in Figure~\ref{fig:extreme2layers}, where there are $k+2$ buyers. Assume that valuations are independently drawn from a uniform distribution $\mathrm{U}[0, 1]$ for all buyers. Then the expected revenue of CWM is
\[ \mathbb{E}[rev^r_{\mathtt{CWM}}] = \frac{5}{24} + 2^{-3-k}\left( k+1 + \frac{\left( -3+2^{2+k} -k \right)k^2}{(1+k)(2+k)} \right) \]

Once $v_1' > \max\{v_2',1/2\}$, the opportunity to check the remaining $k$ buyers is lost. When $k$ becomes larger, this loss cannot be neglected. The expected revenue of CWM in this structure has a measurable gap with $\overline{v} = 1$ when $k\rightarrow \infty$: $\lim_{k\rightarrow\infty} \mathbb{E}[rev^r_{\mathtt{CWM}}] = 5/6$.
% \[ \lim_{k\rightarrow\infty} \mathbb{E}[rev^r_{\mathtt{CWM}}] = \frac{5}{6} \]

However, we can have an alternative mechanism that ignores all buyers in $r_s$ and runs CWM among the rest of the buyers. It is still IC because buyers in $r_s$ are ignored and %invited by the seller and
have no opportunity to manipulate the transaction.
Then, in this example, this alternative mechanism runs Myerson's mechanism among buyers $\{3,4,\dots, k+2\}$, whose expected revenue approaching 1 when $k\rightarrow \infty$.
\end{example}

In Example~\ref{eg:extreme}, we see the power of sacrificing the opportunities of buyers that are closer to the seller. To avoid sacrificing their opportunities completely, a direct idea is slightly increasing the reserve prices of buyers that are close to the seller. For a reported structure profile $r'$, let distance $d_i$ be the length of the shortest diffusion path from $s$ to a valid buyer $i$. We can define a monotone non-increasing \emph{shifting function} $\sigma: \mathbb{N}^* \rightarrow \mathbb{R}_{\geq 0}$, that satisfies $0\leq \sigma(d_i)\leq |\overline{v} - \underline{v}|$. The shifting function is used to define the increment of reserve prices for buyers with different distances. Then, we can define a variant of CWM with shifted reserve prices.

\begin{framed}
\noindent \textbf{CWM with Shifted Reserve Prices (CWM-SRP)}

\noindent \rule{\textwidth}{0.5pt}

\noindent \textsc{Input:} a set of buyers $N$ and their report profile $t'$.
\begin{enumerate}
% \item For each buyer $i\in V(t')$ and her bid $v_i'$, compute her virtual bid $\tilde{v}_i = \phi_i(v_i')$. %Let $\tilde{t} = \{(\tilde{v}_i, r_i')\}_{i\in V(t')}$.
\item Find the potential winner set $W(t')$, represented by $\{ w_1,$ $\dots, w_m \}$, where $w_k\in C(w_{k+1})$ for all $1\leq k<m$. %which satisfies $w_k\in C(w_{k+1})$ for all $1\leq k < m$.
If $W(t') = \emptyset$, then set $\pi_i(t') = 0$, $p_i(t') = 0$ for all $i\in N$ and goto \textsc{Output}.
\item For $k=1$ to $m$: if $v_{w_k}' \geq \phi_{w_k}^{-1}(0) + \sigma(d_{w_k})$, set $\pi_{w_k}(t') = 1$, $p_{w_k}(t') = \max\{p^*_{w_k}(t'), \phi_{w_k}^{-1}(0) + \sigma(d_{w_k})\}$,
and \textsc{Break}.
\item Set $\pi_i(t') = 0$, $p_i(t') = 0$ for all other $i$. %except for winner.
\end{enumerate}
\noindent \textsc{Output:} the allocation $\pi$ and the payment $p$.
\end{framed}

Notice that CWM is a special member of CWM-SRP with the shifting function that always equals to 0. Since a buyer $i$ cannot control $d_i$, she cannot affect the reserve price either. Then, similar to CWM, %each buyer's diffusion will be considered only after she has no chance to win, and a misreported valuation may make her lose the item or pay a deficit. Therefore,
CWM-SRP is still IC and IR. %incentive compatible and individually rational.
In Example~\ref{eg:better}, we can see its improvement in a simple case with 3 buyers.

\begin{example}\label{eg:better}
Consider the case in Example~\ref{eg:extreme} and let $k=1$.
%Consider the structure profile illustrated in Figure~\ref{fig:extreme2layers} with $k=1$. Suppose for all buyers $i$, valuations $v_i$ are drawn independently from a uniform distribution $\mathrm{U}[0, 1]$.
The expected revenue of CWM can be calculated %by the formula given in Example~\ref{eg:extreme}
as $\mathbb{E}[rev^r_{\mathtt{CWM}}] \approx 0.5052 $.
% \[ \mathbb{E}^{\mathtt{CWM}}[rev_s] = \frac{5}{24} + (1+1)\cdot 2^{-3-1} + \frac{2^{-3-1}\left( -3+2^{2+1} -1 \right)}{(1+1)(2+1)} \approx 0.5052 \]

Then, if we set a shifting function as $\sigma = 0.1\times\mathbb{I}(d_i = 1)$, where $\mathbb{I}(\cdot)$ is the indicator function (this means, for buyers in $r_s$, we increase their reserve prices by $0.1$), then the expected revenue of CWM-SRP is $\mathbb{E}[rev^r_{\mathtt{CWM}\text{-}\mathtt{SRP}}] \approx 0.5099 > 0.5052$,
% \begin{align*}
%     \mathbb{E}[rev^r_{\mathtt{CWM}\text{-}\mathtt{SRP}}] & = (1-0.6)\times 0.6^3 + 3\int_{0.6}^1 (1 - u)u^2 \mathrm{d}u
%     + (1-0.6)\times 0.6^2 \\ & \quad + \int_{0.6}^1 (1-u)u \mathrm{d}u + (1-0.5)\times 0.5^3 + 2\int_{0.5}^{0.6} (1-u)u^2 \mathrm{d}u  \\
%     & \approx 0.5099 > 0.5052
% \end{align*}
% \begin{align*}
%     \mathbb{E}[rev^r_{\mathtt{CWM}\text{-}\mathtt{SRP}}] & = 0.4\times 0.6^3 + 3\int_{0.6}^1 (1 - u)u^2 \mathrm{d}u
%     + 0.4\times 0.6^2 \\ & \quad + \int_{0.6}^1 (1-u)u \mathrm{d}u + 0.5^4 + 2\int_{0.5}^{0.6} (1-u)u^2 \mathrm{d}u  \\
%     & \approx 0.5099 > 0.5052
% \end{align*}
which is higher than that of CWM.
\end{example}

In practice, we can set different shifting functions based on the estimations of valuations and structures. %Actually, CWM and CWM-SRP have their own applicable situations.
To provide us with ideas on how to set a shifting function, we evaluate these methods through experiments in the next section.

%%%%%%%%%%%%%%%%%%%%%%%%%%%%%%%%%%%%%%%%%%%%%%%%%%%%%%%%%%%%%%%%%%%%%%%%
\section{Numerical Experiments}\label{sec:exp}
Finally, we conduct experiments to evaluate proposed mechanisms, where the valuations of buyers are drawn independently from the uniform distribution $\mathrm{U}[0, 1]$. What we evaluate and compare include %Mechanisms that are evaluated and compared include
\begin{itemize}
    % \item \textbf{MM in $r_s$}: running Myerson's mechanism among buyers in $r_s$, which has the highest expected revenue if we do not use diffusion auctions,
    % \item \textbf{$k$-PWM}: which gives the optimal revenue can be achieved,
    \item \textbf{IDM}: the first designed diffusion auction in~\cite{li2017mechanism}, %as a representative of existing mechanisms,
    \item \textbf{CWM}: CWM mechanism,
    \item \textbf{CWM-SRP1}: with $\sigma_1(d_i) = 0.1\cdot \mathbb{I}(d_i=1)$, and
    \item \textbf{CWM-SRP2}: with $\sigma_2(d_i) = 0.1(3-d_i) \cdot \mathbb{I}(d_i\leq 2)$.
\end{itemize}

\begin{figure}[t]
    \centering
    \includegraphics[width = \linewidth]{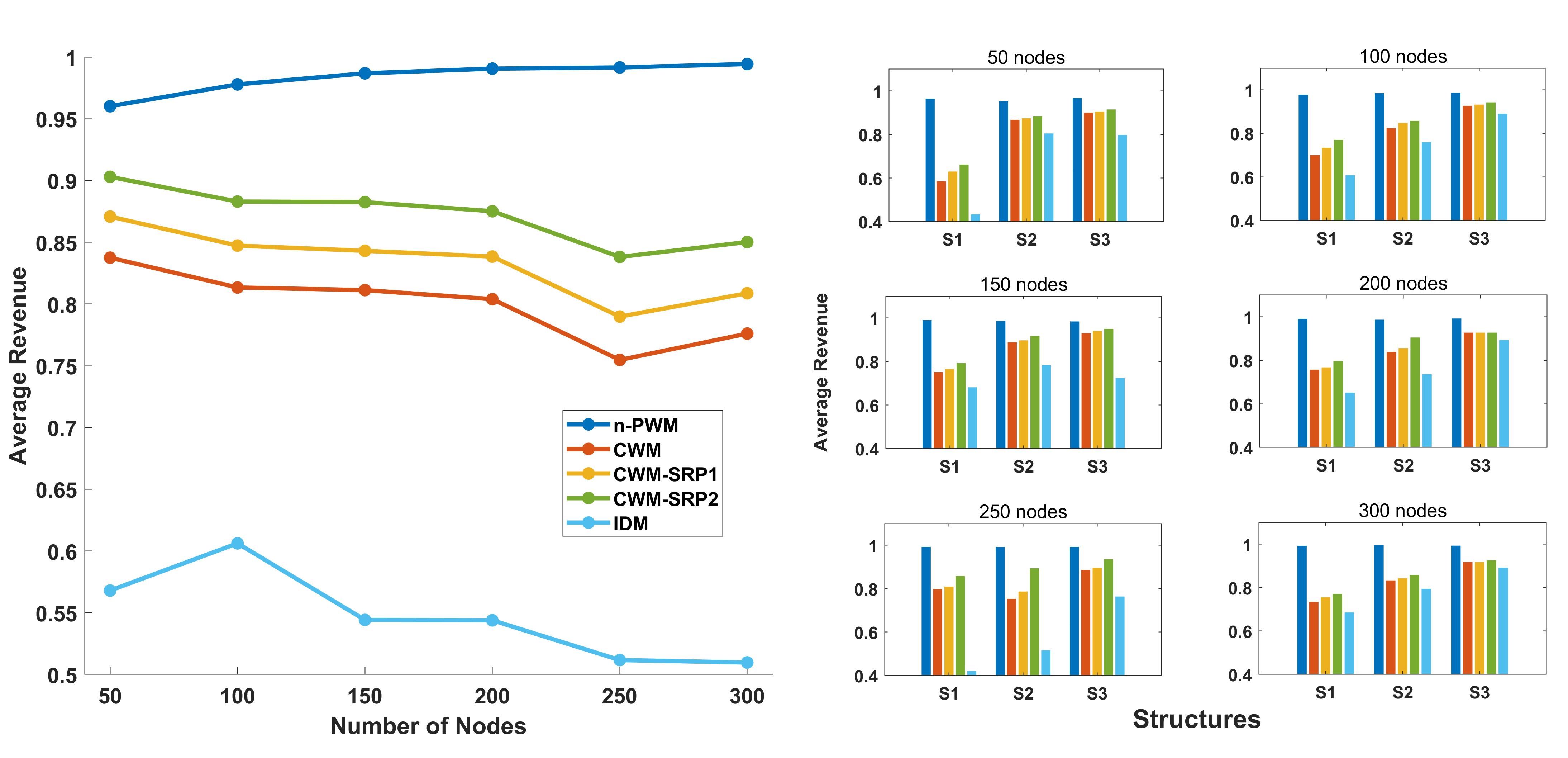}
    \caption{The left graph is the average of the estimated expected revenue over all sampled structures and the right graphs give the estimation of the expected revenue of 3 samples for each number of nodes.} %The four bars with different colors in each subgraph represent the average revenues of $n$-PMW, CWM, CWM-SRP1, and CWM-SRP2, respectively.}
    \label{fig:structure}
\end{figure}

In our experiments, we include the cases of structure profiles with $50\leq n\leq 300$. %Since the number of buyers is too large to transverse all possible structure profiles,
For each number of buyers, %and compute the expected revenue,
%\textcolor{blue}{We test the cases of structure profiles with $50\leq n\leq 300$. For each specific number of nodes,} 
we randomly sample 100 possible structures, and for each structure, we sample 100 groups of valuations and use the average revenue to approximate the expected revenue. To be more practical, we generate the small-world networks, which simulate empirical social networks by a certain category of random graphs~\cite{watts1998smallworld} (it can be done by \textsf{networkx}, a Python programming language package). It has two parameters to sample the structures: the initial degree of nodes, and the probability of randomized re-connection among nodes. We demonstrate the setting with the initial degree to be $2$, and the probability of re-connection to be $0.5$. One can expect a similar result for a different choice of parameters. Figure~\ref{fig:structure} summarizes the %numerical 
results of the average revenue for different numbers of buyers. To see how the expected revenue changes in different structures, we also sampled 3 different structures for each choice of $n$ and %Figure~\ref{fig:structure} shows the
their estimated expected revenue is shown. %for these structures. 
From the results displayed, we can make the following observations.
\begin{itemize}
    \item Compared to the CWM mechanism, both the CWM-SRP mechanisms have significantly improved the expected revenue. In particular, %among the CWM-SRP mechanisms we have tested,
    the CWM-SRP2 always performs the best. Intuitively, with the increase in the number of buyers, the probability that buyers with high valuations are far away from the seller is higher.
    %Intuitively, when the number of buyers is larger, it is more likely to have buyers with high valuations, but they may be away from the seller. At this time, raising the reserve prices of buyers near the seller appropriately can give them a higher probability of getting the item. Hence, the CWM-SRP can effectively improve the expected revenue of the seller.
    \item Different structures have different effects on the expected revenue, but the CWM-SRP2 always performs the best. Actually, for those structures where the CWM mechanism does not perform well, the improvement of the CWM-SRP mechanisms is more significant.
\end{itemize}

Overall, the CWM-SRP mechanisms can greatly improve the expected revenue, and \textbf{when the graph is larger, a more aggressive shifting function performs better}.

\section{Discussion and Future Work}\label{sec:con}
% In this work, we investigate optimal diffusion auctions. %In diffusion auction design, we want to incentivize buyers to propagate the auction information to all their neighbors.
% Based on the idea of potential winner of Myerson's, we give the class of $k$-PWMs, which can be optimal over specific classes of structures. Unfortunately, it also suggests that an optimal mechanism over all structures does not exist. We then propose the CWM and its variants as general mechanisms, which has bounded approximation ratios to optimums.
The ultimate goal in future work is optimizing the expected revenue over the complete induced type profile distribution, including both distributions over valuation and diffusion. The main difficulty is to define a reasonable distribution on different underlying structures with varying number of valid buyers. %(if only focus on distributions over structures with fixed number of valid buyers, $k$-PWMs then already solves the problem).
Possible ways include first giving a distribution on the number of valid buyers and then a distribution on the possible structures, or considering a sufficient large number of potential buyers and the probabilities of the existences between every two buyers. \textbf{In these cases, similar to Myerson's mechanism uses reserve prices to increase expected revenue, $k$-PWM indicates an idea of "reserve structures".} From this aspect, $k$-PWMs may play a more important role in investigating the ultimate goal.

% There are also other future directions worth investigating, e.g., considering a mechanism with a higher approximation ratio than CWM. An upper bound of the approximation ratio is also challenged and worthwhile. In practice, one may find better shifting functions through techniques like machine learning in CWM-SRPs. 
%%%%%%%%%%%%%%%%%%%%%%%%%%%%%%%%%%%%%%%%%%%%%%%%%%%%%%%%%%%%%%%%%%%%%%%%

%%% Use this environment to include acknowledgements (optional).
%%% This will be omitted in doubleblind mode.

\begin{ack}
This work was supported by Shanghai Frontiers Science Center of Human-centered Artificial Intelligence (ShangHAI), MoE Key Laboratory of Intelligent Perception and Human-Machine Collaboration (KLIP-HuMaCo).
\end{ack}

%%%%%%%%%%%%%%%%%%%%%%%%%%%%%%%%%%%%%%%%%%%%%%%%%%%%%%%%%%%%%%%%%%%%%%%%

%%% Use this command to include your bibliography file.

\bibliography{mybibfile}
\end{document}